\documentclass[times,10pt,twocolumn]{article}
\usepackage{latex8}
\usepackage{times}
\usepackage{graphicx, url, setspace}
\usepackage{paralist}
\usepackage[top=1in, bottom=.93in, left=1in, right=1in]{geometry}
\usepackage{algorithmic}
\usepackage{color}
\usepackage[small,compact]{titlesec}
\usepackage{amssymb, amsmath, latexsym, amsthm}
\usepackage{lscape}
\usepackage[ruled,vlined,boxed]{algorithm2e}

\newcommand{\INT}{{\sf int}}

\newtheorem{theorem}{Theorem}[section]
\newtheorem{definition}{Definition}[section]
\newtheorem{assumption}{Assumption}[section]
\newtheorem{corollary}{Corollary}[section]
\newtheorem{claim}{Claim}[section]

\newcommand{\ANON}{{\small \sf AND$\overline{\mbox{a}}$NA}}
\newcommand{\ANONTITLE}{{\Large \sf AND$\overline{\mbox{a}}$NA}}
\newcommand{\ANONSMALL}{{\footnotesize \sf AND$\overline{\mbox{a}}$NA}}

\begin{document}

\title{\ANONTITLE{}: Anonymous Named Data Networking Application}
\author{
    Steven DiBenedetto\\
    Colorado State University\\
    {\it dibenede@cs.colostate.edu}
  \and
    {Paolo Gasti \hspace*{0.7cm} Gene Tsudik} \\
    University of California, Irvine\\
    {\it \{pgasti,gtsudik\}@uci.edu}
   \and
    Ersin Uzun\\
    Palo Alto Research Center\\
    {\it ersin.uzun@parc.com}
}
\date{}
\maketitle
\begin{abstract}
Content-centric networking --- also known as information-centric
networking (ICN) --- shifts emphasis from hosts and interfaces (as in
today's Internet) to data. Named data becomes addressable and
routable, while locations that currently store that data become
irrelevant to applications.

Named Data Networking (NDN) is a large
collaborative research effort that exemplifies the
content-centric approach to networking.
NDN has some innate privacy-friendly features, such as lack of source
and destination addresses on packets. However, as discussed in this
paper, NDN architecture prompts some privacy concerns mainly stemming
from the semantic richness of names. We examine privacy-relevant
characteristics of NDN and present an initial attempt to achieve
communication privacy. Specifically, we design an NDN
add-on tool, called \ANONSMALL{}, that borrows a number of features
from Tor. As we demonstrate via experiments, it provides comparable
anonymity with lower relative overhead.

\end{abstract}

\section{Introduction}
Although the Internet, as a whole, is a huge global success story, it
is showing clear signs of age. In the 1970s, when core ideas
underlying today's Internet were developed, telephony was the only
example of effective global-scale communications. Thus, while the
communication solution offered by the Internet's TCP/IP suite was
unique and ground-breaking, the communication paradigm it focused on
was similar to that of telephony: a point-to-point conversation
between two entities. The communication world has changed
dramatically since then and today's Internet has to accommodate:
information-intensive services, exabytes of content created and
consumed daily over the Web as well as a menagerie of mobile devices
connected to it. To keep pace with these changes and move the
Internet into the future, a number of research efforts to design new
Internet architectures have taken off in the last few years.

Named-Data Networking (NDN) \cite{NDN} is one such effort that
exemplifies the content-centric approach
\cite{gritter2001architecture, Jacobson2009, koponen2007data} to
networking. NDN names content instead of locations (i.e.,
hosts or interfaces) and thus transforms content into a first-class
entity. NDN also stipulates that each piece of content must be signed
by its producer. This allows decoupling of trust in content from trust
in the entity that might store and/or disseminate that content. These
NDN features facilitate automatic caching of content to optimize
bandwidth use and enable effective simultaneous utilization of
multiple network interfaces.

However, NDN introduces certain challenges that must be addressed in
order for it to be a serious candidate for the future Internet
architecture. One major argument for a new architecture is the
inadequate level of security and privacy in today's Internet.
We view anonymity as being a critical feature in any new network
architecture. It helps people overcome communication restrictions and
boundaries as well as evade censorship. In addition, some
applications (e.g., e-cash or anonymous publishing) can be
successfully deployed only if the underlying network allows users to
hide their identity \cite{chaum85}. Even if end-users do not care
about anonymity with respect to services they access, they might
still want to hide their activities from employers, governments and
ISPs, since those might censor, misuse or accidentally leak sensitive
information \cite{feigenbaum2007model}.

Lack of source/destination addresses in NDN helps
privacy, since NDN packets carry information only about \emph{what}
is requested but not \emph{who} is requesting it. However, a closer look
reveals that this is insufficient. In particular, NDN design introduces
three important privacy challenges:
\begin{compactenum}
\item {\bf Name privacy}: NDN content names are incentivized to be
    semantically related to the content itself. Similar to HTTP headers,
    names reveal significantly more information about content than
    IP addresses. Moreover, an observer can easily determine when two
    requests refer to the same (even encrypted) content.
\item {\bf Content privacy}: NDN allows any entity that knows
    a name to retrieve corresponding content. Encryption in
    NDN is used to enforce access control and is not
    applied to publicly available content. Thus, consumers wanting to
    retrieve public content cannot rely on encryption to hide what they access.
\item {\bf Cache privacy}: as with current web proxies, network
    neighbors may learn about each others' content access using
    timing information to identify cache hits.
\item {\bf Signature privacy}: since digital signatures in NDN
    content packets are required to be publicly verifiable, identity of a content
    signer may leak sensitive information.
\end{compactenum}
In this paper, we attempt to address these challenges. We present an
initial approach, called \ANON{} that can be viewed as an adaptation of onion routing to
NDN. Our approach is in-line with NDN principles.
It is designed to take advantage of NDN strengths and work around its
weaknesses. We optimized \ANON{} for small- to medium-size interactive
communication -- such as web-browsing and instant messaging --
that are characterized by moderate amounts of low-latency
traffic \cite{Callahan:2010:LVH:1889324.1889347}.

We provide a security analysis of the proposed approach under a
realistic adversarial model. Specifically, we define anonymity and
unlinkability under this security model and show that these
properties hold for \ANON{}. Moreover, \ANON{} is secure with fewer
anonymizing router hops than Tor. We prototyped \ANON{} and assessed
its performance via experiments over a network testbed. Results show
that \ANON{} introduces less overhead than Tor, especially, for
anticipated traffic patterns.

We believe that this work is both timely and important. The former --
because of the recent surge of interest in content-centric networking
and NDN being a good example of this paradigm. (Also, while NDN is
sufficiently mature to have a functional prototype suitable for
experimental use, it is still at an early enough stage to be open to
change.) The latter -- because it represents the first attempt to
identify and address privacy problems in a
viable candidate for the future Internet architecture.

\medskip

Before discussing details of our approach, we
present further motivation for this work.

\medskip\noindent{\bf Why NDN?}
There are multiple efforts to develop new content-centric
architectures and NDN is only one of those. We focus on NDN because
it stands out in several aspects. First, it combines some
revolutionary ideas about content-based routing that have attracted
considerable attention from the networking research community.
Second, it builds upon an open-source code-base called CCNx
\cite{ccnx}, that is led and continuously maintained by an industrial
research lab (PARC). At the time of this writing (summer 2011), NDN
is one of the very few content-centric architectural proposals with a
reasonably mature prototype available to the research
community.\footnote{We are aware of only two other content-centric
architecture proposals -- \cite{netinf} and \cite{pursuit} -- that
have public prototypes.} Third, NDN is one of only four projects
selected by NSF Future Internet Architectures (FIA) program
\cite{FIA}.

On the other hand, NDN is an on-going research project and is thus
subject to continuous change. However, we believe that it represents
a good example of content-centric networking design and at least some
of its concepts will influence the future of networking. More
importantly, ideas, techniques and analysis discussed in this paper
are not specific, or limited to, NDN; they are applicable to a wide
range of designs, including host-, location- and content-addressable
networks.

\medskip\noindent{\bf Approach. }
NDN follows the proven design principle of IP and claims to be the
``thin waist'' of the communications protocol stack. Thus, pushing
security or privacy services (that are not critical for all types of
communication) into this thin waist would contradict its design
principle. Consequently, as in the case of IP, we believe that
privacy tools should run on top of NDN. Looking at privacy and
anonymity techniques in today's Internet, one well-established
approach is an overlay anonymization network, exemplified by
Tor~\cite{tor}. Tor and its relatives employ layers of concentric
encryption and intermediate nodes responsible for peeling off layers
as packets travel through the overlay. This  is commonly
referred to as onion routing. Our approach falls into roughly the same
category. However, as we discover and discuss in this paper, the task
of adapting an anonymization overlay approach to NDN is not as simple
as it might initially seem.

\medskip\noindent{\textbf{Scope. }}
The primary focus of this paper is privacy. Security and other
features of NDN are taken as given without justifying their
existence. A number of important NDN-related security topics are out
of scope of this paper, including: trust management, certification
and revocation of credentials as well as routing security.

\medskip\noindent{\textbf{Organization. }}
We start with NDN overview and privacy analysis in Section
\ref{sec:netmodel}. Section \ref{sec:relwork} summarizes related
work, followed by the description of \ANON{} in Section
\ref{sec:andora-descr}. Section \ref{sec:security_analysis}
introduces a formal model for provable anonymity and security
analysis of \ANON{}. Implementation details and performance
evaluation results are discussed in Section \ref{sec:implem}. The
paper concludes in Section \ref{sec:conclusions}.

\section{NDN Overview}
\label{sec:netmodel}

NDN \cite{NDN} is a communication architecture based on named
content.\footnote{Note that we use the terms "content" and "data"
interchangeably throughout this paper.} Rather than addressing
content by its location, NDN refers to it by name. Content name is
composed of one or more variable-length components that are opaque to
the network. Component boundaries are explicitly delimited by ``{\tt
/}''. For example, the name of a CNN news content might be: {\small
\tt /ndn/cnn/news/2011aug20}. Large pieces of content can be split
into fragments with predictable names: fragment $137$ of a YouTube
video could be named: {\small \tt
/ndn/youtube/videos/video-749.avi/137}.

Since the main abstraction is content, there is no
explicit notion of ``hosts'' in NDN. (However, their existence is assumed.)
Communication adheres to the {\em pull} model: content is delivered
to consumers only upon explicit request. A consumer requests content
by sending an {\em interest} packet. If an entity (a router or a host)
can ``satisfy'' a given interest, it returns the corresponding {\em content}
packet. Interest and content are the only types of packets in NDN. A
content packet with name X in NDN is {\bf never} forwarded or routed
unless it is preceded by an interest for name X.\footnote{Strictly speaking,
content named $X'\neq~X$ can be delivered in response to an interest
for $X$ but only if $X$ is a prefix of $X'$. As an example, the full name of each content packet contains the hash of that content; however, this hash value is usually not known to consumers and is typically omitted from interests.}

When a router receives an interest for name X and there are no
pending interests for the same name in its PIT (Pending
Interests Table), it forwards this interest to the next hop according
to its routing table. For each forwarded interest, a router stores
some state information, including the name in the interest and the
interface on which it was received. However, if an interest for
X arrives while there is an entry for the same name in the
PIT, the router collapses the present interest (and any subsequent
ones for X) storing only the interface on which it was received. When
content is returned, the router forwards it out on all interfaces where
an interest for X has been received and flushes the corresponding
PIT entry. Note that, since no additional information is needed to deliver
content, an interest does not carry a source address. More detailed
discussion of NDN routing can be found in \cite{Jacobson2009}.

In NDN, each network entity can provide content caching, which is
limited only by resource availability. For popular content, this
allows interests to be satisfied from cached copies distributed over
the network, thus maximizing resource utilization. NDN
deals with content authenticity and integrity by making
digital signatures mandatory on all content packets. A signature
binds content with its name, and provides origin authentication no
matter how or from where it is retrieved.
NDN calls entities that publish
new content {\em producers}. Whereas, as follows from the above
discussion, entities that request content are called {\em consumers}.
(Consumers and producers are clearly overlapping sets.)
Although content signature verification is optional in NDN, a signature
must be verifiable by any NDN entity. To make this
possible, content packets carry additional metadata, such as the
ID of the content publisher and information on locating the
public key needed for verification. Public keys are treated as regular
content: since all content is signed, each public key content
is effectively a ``certificate''. NDN does not mandate any
particular certification infrastructure, relegating trust management
to individual applications.
Private or restricted content in NDN is protected via encryption by the
content publisher. Once content is distributed unencrypted, there is
no mechanism to apply subsequent encryption. Specific applications
may provide a means to explicitly request encryption of content by
publishers. However, NDN does not currently allow consumers to
selectively conceal content corresponding to their interests.

From the privacy perspective, lack of source and destination
addresses in NDN packets is a clear advantage over IP. In practice,
this means that the adversary that eavesdrops on a link close to a
content producer can not immediately identify the consumer(s) who
expressed interest in that content. Moreover, two features of
standard NDN routers: (1) content caching and (2) collapsing of
redundant interests, reduce the utility of eavesdropping near a
content producer since not all interests for the same content
reach its producer.

On the other hand, NDN provides no protection against an adversary
that monitors local activity of a specific consumer. As most
content names are expected to be semantically relevant to content
itself, interests can leak a lot of information about the content
they aim to retrieve. To mitigate this issue, NDN allows the use of ``encrypted names'', whereby a producer encrypts the tail-end
(a few components) of a
name~\cite{Jacobson2009}.
\footnote{For example, a name such as: {\tt
/ndn/xerox/parc/Alice/family/photos/Hawaii} might be
replaced with {\tt /ndn/xerox/parc/Alice/\textbf{encrypted-part}}.}
However, this simple approach does not provide much privacy: the
adversary can link multiple interests for the same content -- or
those sharing the same name prefix -- issued by different consumers.
Moreover, an adversary can always replay an interest to see what
(possibly cached) content it returns, even if a name of content is
not semantically relevant.

\section{Related Work}
\label{sec:relwork}
The goal of anonymizing tools and techniques is to decouple actions from entities
that perform them. The most basic approach to anonymity is
to use a trusted anonymizing proxy. A proxy is typically interposed
between a sender and a receiver in order to hide identity of the former
from the latter. The Anonymizer \cite{anonymizer} and Lucent
Personalized Web Assistant \cite{gabber99} are examples of
this approach. While relatively efficient, it is
susceptible to a (local) passive adversary that monitors all
proxy activity. Also, a centralized proxy necessitates centralized (global)
trust and represents a single point of  failure.

A more sophisticated decentralized approach is used in mix networks
\cite{mixnets}. Typically, a mix network achieves anonymity by
repeatedly routing a message from one proxy to another, such that the
message gradually loses any relationship with its originator.
Messages must be made unintelligible to potentially untrusted
intermediate nodes. Chaum's initial proposal
\cite{mixnets} defines an anonymous email system, wherein a sender
envelops a message with several concentric layers of public key encryption. The
resulting message is then forwarded to a sequence of {\em mix}
servers, that gradually remove one layer of encryption at a time and
forward the message to the next mix server.

Subsequent research generally falls into two classes:  delay-tolerant
applications (e.g. email, file sharing) and real-time  or low-latency
applications (e.g. web browsing, VoIP, SSH). These two classes
achieve different tradeoffs between performance (in terms of latency
and bandwidth) and anonymity. For example, Babel \cite{babel},
Mixmaster \cite{mixmaster} and Mixminion \cite{mixminion} belong to
the first category. Their goal is to provide anonymity with respect
to the {\em global eavesdropper} adversary. Each mix
introduces spurious traffic and randomized traffic delays in order
to inhibit correlation between  input and output traffic. However,
unpredictable traffic characteristics and high delays make
these techniques unsuitable for many applications.

Low-latency anonymizing networks are at the other end of the
spectrum. They try to minimize extra latency by forwarding traffic as
fast as possible. Because of this, strategies used in anonymization
of delay-tolerant traffic -- batching (delaying) and re-ordering of
traffic in mixes, as well as  introduction of decoy traffic --- are
generally not applicable. For example, \cite{serjantov2003} shows how
traffic patterns can be used for de-anonymization in low-latency
anonymity systems. Notable low-latency tools are summarized below.

Crowds \cite{crowds} is a low-latency anonymizing network for HTTP
traffic. It differs from traditional mix-based approaches as it lacks
layered encryption. For each message it receives, an anonymizer
probabilistically chooses to either forward it to a random next hop
within the Crowds network or deliver  it to its final destination.
Since messages are not encrypted, Crowds is vulnerable to local
eavesdroppers and predecessor attacks \cite{wright2004predecessor}.

Morphmix \cite{morphmix1,morphmix2} is a fully distributed
peer-to-peer mix network that uses layered encryption. Unlike Crowds,
it does not require a lookup service to keep track of all
participating nodes. Senders selects the first anonymizer and each
anonymizer along an ``anonymous tunnel'' picks the next hop to
dynamically build tunnels. Tarzan \cite{tarzan} is another fully
distributed peer-to-peer mix network. It builds a universally
verifiable set of neighbors (called mimics) for every node to keep
track of other other Tarzan participants. Every node selects its
mimics pseudo-randomly.

Tor \cite{tor} is the best-known and most-used low-latency
anonymizing tool. It is based on onion routing and layered
encryption. Tor uses a central directory to locate participating
nodes and requires users to build a three-hop anonymizing circuit by
choosing three random nodes. The first is called the {\em guard}, the
second -- the {\em middle}, and the third --- {\em exit} node. Once
set up, each circuit in Tor lasts about 10 minutes. For better
performance, bandwidth available to nodes is taken into account
during circuit establishment and multiple TCP connections are
multiplexed over one circuit. Communication between Tor nodes is
secured via SSL.  However, Tor does not introduce any decoy traffic
or randomization to hide traffic patterns. Another anonymization tool, I2P
\cite{I2P}, adopts many ideas of Tor, while using a distributed
untrusted directory service to keep track of its participants. I2P
also replaces Tor's circuit-switching operation with packet-switching
to achieve better load balancing and fault-tolerance.

A consumer privacy technique for Information-Centric Networks (ICNs) 
is proposed in \cite{Arianfar11}. Instead of using encryption,
it leverages cooperation from content producers and requires them
to mix sensitive information with so-called ``cover'' content. 
This approach requires producers to cooperate and store a
large amount of cover traffic. It also does not provide consumer-producer
unlinkability or protection against malicious producers.

Telex \cite{Wustrow11} is an alternative to mix networks designed to 
evade state-level censorship. It 
uses steganographic techniques to hide messages in SSL handshakes.
Users connect to innocuous-looking unblocked websites through SSL.
Sympathetic ISP-s that forward user's traffic recover hidden
messages and deliver them to the intended destination. While novel,
this approach presents significant deployment challenges and requires
support from the network infrastructure. Furthermore, the threat
model in Telex is quite different from that of the 
other anonymizing tools presented above.
Moreover, established TCP fingerprinting techniques can easily detect
differences between a Telex station and a censored website.
Another analogous technique -- called  Cirripede \cite{cirripede} -- was recently
proposed.

\section{AND$\overline{\mbox{a}}$NA}

\label{sec:andora-descr}

\ANON{} is a onion routing overlay network, built on top of NDN, that
provides privacy and anonymity to consumers. In particular, \ANON{}
prevents adversaries from linking consumers with the content they are
retrieving. Following the terminology introduced in \cite{crowds},
\ANON{} provides {\em beyond suspicion}\footnote{For any packet observed by 
the adversary, an entity is considered {\em beyond suspicion} if it is as likely to be the 
sender of this packet as any other entity.} degree of anonymity to its users.

\ANON{} uses multiple concentric layers of encryption and routes
messages from consumers through a chain of at least two onion
routers. Each router removes a layer of encryption and forwards the
decrypted messages to the next hop. Due to its low-latency focus,
\ANON{} does not guarantee privacy in presence of a global
eavesdropper. However, since it is geared for a world-wide (or at
least geographically distributed) network spanning a multitude of
administrative domains, the existence of such an adversary is
unlikely. For this reason, we restrict the adversarial capabilities
to eavesdropping on, injecting, removing or modifying messages on a
subset of available links. An adversary can compromise NDN routers
and \ANON{} nodes at will. Nonetheless, consumers benefit from
anonymity as long as they use at least one non-compromised \ANON{}
node. Details of our adversarial model and formal privacy guarantees
are discussed in  Section \ref{sec:security_analysis}.

\subsection{Design}
We now present two techniques --- {\em asymmetric} and {\em
session-based} --- that provide privacy and anonymity for NDN
traffic. %
Traffic is routed through {\em ephemeral circuits}, that are defined
as a pair of distinct anonymizing routers (ARs).  An AR is a NDN node
(e.g. a router or a host) that chooses to be part of \ANON. An ephemeral circuit
transports only one (or only a few) encrypted interest(s). It disappears either
when the corresponding content gets delivered, or after a
short timeout (hence ``{\em ephemeral}''). A timeout interval is
needed so that the consumer can re-issue the same encrypted interest
in case of packet loss. We refer to the first AR as {\em entry
router} and the second -- as {\em exit router}. They must not belong
to the same administrative domain and must not share the same name prefix.
Optionally, consumers can select ARs according to some parameters,
such as advertised bandwidth, availability or average load.
As pointed out in \cite{bauer2007low,MurdochW08}, there is a well know
natural tension between non-uniform (i.e. performance-driven) choice
of routers and anonymity. Consumers should consider this when
selecting ARs.

To build an ephemeral circuit, a consumer retrieves the list of ARs
and corresponding public keys. Although we do not mandate any
particular technique, a consumer can retrieve this list using, e.g.,
a directory service \cite{tor} or a decentralized (peer-to-peer)
mechanism.
AR public keys can be authenticated using
decentralized techniques (such as web-of-trust \cite{abd97}) or a PKI
infrastructure.\footnote{Note that implicit replication implemented
through caching allows the construction of a directory system with
better resilience against denial-of-service (DoS) attacks than IP.}

A prospective AR joins \ANON{} by advertising its public key,
together with its identity defined as: namespace, organization and
public key fingerprint. An AR also publishes auxiliary  information,
such as total bandwidth, average load, and uptime.

As mentioned earlier, both interest and content packets leak
information. Even if names in interests are hidden, three components
of content packets --- signatures, names and content itself ---
contain potentially sensitive information. Of course, content
producers could simply generate a new key-pair to sign each content
packet. This would be impractical, since high costs of key generation
and distribution would make it difficult for consumers to
authenticate content. (Note that key-evolving schemes
\cite{BellareM99} do not help, since verification keys generally
evolve in a way that is predictable to all parties, including the
adversary.)
Alternatively, the original content signature could be replaced with
that generated by an AR. However, this would preclude end-to-end
content verifiability and thus break the NDN trust model.

For this reason, \ANON{} implements encrypted encapsulation of
original content, using two symmetric keys securely distributed by
the consumer to the ARs during setup of the ephemeral circuit. Upon
receiving a content packet, the exit router encrypts it, together
with the original (cleartext) name and signature, under the first key
provided by the consumer. Then, treating the ciphertext as payload
for a new content packet, the exit router signs and sends it to
the entry router. The latter strips this signature and the name
and encrypts the remaining ciphertext under the second symmetric key
provided by the consumer. Next, it forwards the ciphertext with the
original encrypted name and a fresh (its own) signature. After
decrypting the payload, the consumer discards the signature from the
entry router and verifies the one from the content producer.

Because decryption is deterministic, an encrypted interest sent to an
AR always produces the same output. Since ARs are a public resource,
the adversary can use them to decrypt previously observed interests.
It can thus observe the corresponding output and correlate
incoming/outgoing interests. This is a well-known attack and there
are several ways to mitigate it, such as encrypted channels between
communicating parties \cite{tor} and mixing (for delay-tolerant
traffic) \cite{babel}. However, such techniques tend to have
significant impact on computational costs and latency. Instead, we use
standard NDN features of interest aggregation and caching to prevent
such attacks, as described next.

In NDN, a router (not just an AR) that receives duplicate interests
collapses them. An interest is considered a {\em duplicate}, if it
arrives while another interest referring to the same content has not
been satisfied. Also, if the original interest has been satisfied
and the corresponding content is still in cache, a new interest
requesting the same piece of data is satisfied with cached content.
In this case, the router does not forward any interests. Therefore,
the adversary must wait for the expiration of cached content.

As part of \ANON, the consumer includes its current timestamp within
each encryption layer. ARs reject interests with timestamps outside a
pre-defined time window. Thus, consumers need to be loosely synchronized
with ARs that must reserve at least $(rate \times window)$ of
cache, where $rate$ is the router's wire-rate and $window$ is the
interval within which interests are accepted. In this way, if an
interest is received multiple times by an AR (e.g.~in case  of loss
of the corresponding data packet between the AR and the consumer),
the AR is able to satisfy it using its cache.

The encryption algorithm used by consumers to conceal names in
interests must be secure against adaptive chosen ciphertext (CCA)
attacks.\footnote{Technically, in order to guarantee correctness an encryption scheme suitable for \ANONSMALL{} must also be robust \cite{AbdallaBN10}. However, since CCA-secure encryption schemes used in practice  are also robust, we omit this requirement in the rest of the paper.}
CCA-security \cite{BellareN08} implies, among other things,
probabilistic encryption and non-malleability. The former prevents
the adversary from determining whether two encrypted interests
correspond to the same unencrypted interest. Whereas, the latter
implies that the adversary cannot modify interests to defeat the
mechanism described above.

We now describe two flavors of anonymization protocols:
asymmetric and session-based. In order to allow efficient routing of
interest packets, the encrypted component is encoded at the end of
the name with both flavors.

\medskip
\noindent {\bf Asymmetric:} To issue an interest, a consumer
selects a pair of ARs and uses their public keys to encrypt
the interest, as described above and in Algorithm
\ref{alg:onion_user}. A consumer also generates two symmetric keys:
$k_1$ and $k_2$ that will be used to encrypt the content packet on
the way back. We use $\mathcal{E}_{pk}(\cdot)$ and
$\overline{\mathcal{E}_k}(\cdot)$ to denote (CCA-secure) public key and
symmetric encryption schemes, respectively.

To account for the delay due to extra hops needed to reach the second
AR (and reduce the number of discarded interests), a consumer adds
half of the estimated round trip time (RTT) to the innermost
timestamp. Each AR removes the outermost encryption layer, as
detailed in Algorithm \ref{alg:onion_router_int}. Since
$\mathcal{E}_{pk}(\cdot)$ is CCA-secure, the decryption process fails
if the ciphertext has been modified in transit or was not encrypted
under the AR's public key. Content corresponding to the encrypted
interest is encrypted on the way back, as detailed in Algorithm
\ref{alg:onion_router_data}, using $\overline{\mathcal{E}_k}(\cdot)$
and symmetric keys supplied by the consumer.

\begin{algorithm}[t!]
\footnotesize
\caption{Encrypted Interest
Generation}\label{alg:onion_user}
\SetKwInOut{Input}{input}\SetKwInOut{Output}{output}
\Input{Interest $\INT$; \ \  Set of $\ell$ ARs and their keys: $\mathcal{R} =
\{(\mbox{AR}_i, pk_i)\;|\; 0<i\leq~\ell\;, pk_i\in \mathcal{PK}\}$}
\Output{Encrypted interest $\INT_{pk_i, pk_j}$; \ symmetric keys $k_1,  k_2$}
\begin{algorithmic}[1]
\STATE Select $(\mbox{AR}_i, pk_i),(\mbox{AR}_j, pk_j)$ from $\mathcal{R}$
\IF{AR$_i = $ AR$_j$ \OR AR$_i$, AR$_j$ are from same
organization \OR \\ AR$_i$, AR$_j$ share the same name prefix}
\STATE Go to line 1
\ENDIF
\STATE $k_1 \leftarrow \{0,1\}^\kappa$ ; \ \ $k_2 \leftarrow \{0,1\}^\kappa$
\STATE $eint = \mbox{AR}_2/\mathcal{E}_{pk_j}(\INT \ | \ k_2 \ | \ curr\_timestamp + RTT/2)$
\STATE $eint = \mbox{AR}_1/\mathcal{E}_{pk_i}(eint \ | \ k_1 \ | \ curr\_timestamp)$
\STATE Output $eint, k_1, k_2$
\end{algorithmic}
\end{algorithm}

\begin{algorithm}[t!]
\footnotesize \caption{AR Handling of Encrypted
Interests}\label{alg:onion_router_int}
\SetKwInOut{Input}{input}\SetKwInOut{Output}{output}
\Input{Encrypted Interest $\INT_{pk_i, pk_j}$, where $pk_i,pk_j \in \mathcal{PK}\  \cup \{\perp\}$ \ \ \
(where ``$\perp$'' denotes ``no encryption'')}

\Output{Interest $\INT_{pk_j}; \ $ symmetric key $k_1$}

\begin{algorithmic}[1]
\footnotesize \STATE $(\INT_{pk_j}, k_1, timestamp) =
\mathcal{D}_{sk_i}(\INT_{pk_i, pk_j})$ \IF{Step 1 fails \OR
$timestamp$ is not current} \STATE Discard $\INT_{pk_i, pk_j}$ \ELSE
\STATE Save tuple  $(\INT_{pk_i, pk_j},\INT_{pk_j}, k_1)$ to internal state
\STATE Output
$\INT_{pk_j}, k_1$ \ENDIF
\end{algorithmic}
\end{algorithm}

\begin{algorithm}[t!]
\footnotesize \caption{AR Content
Routing}\label{alg:onion_router_data}
\SetKwInOut{Input}{input}\SetKwInOut{Output}{output}
\Input{Content: $data_{k_2}$ in response to $\INT_{pk_j}$, where  $pk_j \in \mathcal{PK}\  \cup \{\perp\}$}
\Output{Encrypted data packet $data_{k_1, k_2}$}
\begin{algorithmic}[1]
\footnotesize
\STATE Retrieve tuple $(\INT_{pk_i, pk_j},\INT_{pk_j}, k_1)$ from internal state\\ where name in $\INT_{pk_2}$ matches that in $data_{k_2}$
\STATE {\bf if} $k_2 \neq\ \perp$ {\bf then} Remove signature and name from $data_{k_2}$
\STATE Create new empty data packet $pkt$
\STATE Set name on $pkt$ as the name on $\INT_{pk_i, pk_j}$
\STATE Set the data in $pkt$ as $\overline {\mathcal{E}_{k_1}}(data_{k_2})$
\STATE Sign $pkt$ with AR's key
\STATE Output $pkt$ as $data_{k_1, k_2}$
\end{algorithmic}
\end{algorithm}

\medskip
\noindent {\bf Session-based Variant. } This variant aims to reduce
(amortize) the use of public key encryption thus lowering the
computational cost and ciphertext size. Before sending any interests
through ephemeral circuits, a consumer (Alice) establishes a shared
secret key with each selected AR. This is done via a 2-packet
interest/content handshake. We do not describe the details of
symmetric key setup, since there are standard ways of doing it.
We provide two options: one using Diffie-Hellman key exchange
\cite{diffie1976new}, and the other -- using SSL/TLS-style protocol
whereby Alice encrypts a key for $AR_i$. Once a symmetric key
$k_{ai}$ is shared with $AR_i$, Alice can establish any number of
ephemeral circuits using it as either first or second AR hop. Also at
setup time, Alice and $AR_i$ agree on session identifier value --
$sid_{ai}$ -- that is included (in cleartext) in subsequent interests
so that $AR_i$ can identify the appropriate entry for Alice and
$k_{ai}$.

The main advantage of the session-based approach is better
performance: both consumers and routers only perform symmetric
operations after initial key setup. However, one drawback is that,
since the session identifier $sid$ is not encrypted, packets
corresponding to the same $sid$ are easily linkable.

We note that our design neither encourages nor prevents consumers
from mixing asymmetric {\em and} session-based variants for the same
or different ephemeral circuits.

\subsection{System and Security Model}
In order for our discussion to relate to prior work, we
use the notion of ``indistinguishable configurations'' from the framework
introduced in \cite{feigenbaum2007model}; the actual definitions are
in Section \ref{sec:security_analysis}.

Our security analysis considers the worst case scenario, i.e.,
interests being satisfied by the content producer rather than a
router's cache. While, in normal conditions, encrypted interests are
satisfied by caches only in case of packet loss, fully decrypted
interests may not have to reach to content producers. A system secure
in case of cache misses is also secure when interests are satisfied
by content cached at routers along the way. (Recall that,
when an interest is satisfied by a router's
cache, it is not forwarded any further.) This limits the adversary's
ability to observe interests in transit.

\medskip
\noindent {\bf Adversary Goals and Capabilities. } The goal of an
adversary is to link consumers with their actions. In particular, it
may want to determine what content is being requested by a particular
user and/or which users are requesting specific content. A somewhat
related goal is determining which cache (if any) is satisfying a
consumer's requests.
Our adversary is local and active: it controls only a subset of
network entities and can perform any action usually allowed to such
entities. Moreover, it is capable of selectively compromising
additional network entities according to its local information.

Our model allows the adversary to perform the following actions:
\begin{compactitem}
\item {\bf Deploy compromised routers}: \ANON{} is an open network,
    therefore an adversary can deploy compromised anonymizers and
    regular routers. As such, routers may exhibit malicious behavior
    including injection, delay, alteration, or drop traffic.
\item {\bf Compromise existing routers}: An adversary can select any
    router (either ARs or regular routers) in the network and
    compromise it. As a result, the adversary learns all the private
    information (e.g.~decryption keys, pending decrypted interests,
    cache content, etc.) of such router.
\item {\bf Control content producers}: Content producers are not part
    of  \ANON{}. As such, the network has no control over them. An
    adversary can compromise existing content producers or deploy
    compromised ones and convince users to pull content from them. We
    also assume that the content providers are publicly accessible,
    and therefore the adversary is able to retrieve content from
    them.
\item {\bf Deploy compromised caches}: Similarly to compromised
    content producers, an adversary can compromise routers' cache or
    deploy its own caches. The behavior of a compromised cache
    includes monitoring cache requests and replying with corrupted
    data.
\item {\bf Observe and replay traffic}: An adversary can tap a link
    carrying anonymized traffic. By doing this it learns, among other
    things, packet contents and traffic patterns. The traffic
    observed by an adversary can be replayed by any compromised
    router.

\end{compactitem}
\medskip
An adversary can iteratively compromise entities of its choice, and
use the information it gathers to determine what should be
compromised next. In order to make our model realistic, the time
required by an adversary to compromise or deploy a router, a cache or
a content producer is significantly higher that the round-trip time
(RTT) of an anonymized interest and corresponding data. This implies
that all the state information recovered from a newly compromised
router only refers to packets received {\em after} the adversary
decides to compromise such router.

A powerful class of attacks against anonymizing networks is called
fingerprinting~\cite{mittalstealthy,shmatikov2006timing}.
Inter-packet time intervals are usually not hidden in low latency
onion routing networks because packets are dispatched as quickly as
possible. This behavior can be exploited by an adversary, who can
correlate inter-packet intervals on two links and use this
information to determine if the observed packets belong to the same
consumer~\cite{shmatikov2006timing}. This class of attacks is
significantly harder to execute on \ANON{} because of the nature of
ephemeral circuits and because of the use of caches on routers.
Ephemeral circuits do not allow the adversary to gather enough
packets with uniform delays since they are used to transport only one
or a very small number of interests and corresponding data. Active
adversaries who can control the communication link of a content
provider can add measurable delays to some of the packets in order to
identify consumers. However, consumers may be able to retrieve the
same content through caches making such attack ineffective.
Throughput fingerprinting consists in measuring the throughput of the
circuit used by a consumer to identify the slowest anonymizer in the
consumer's circuit \cite{mittalstealthy}. Throughput fingerprinting
is difficult to perform in \ANON{} since each ephemeral circuit does
not carry enough information to mount an attack. In particular, the
authors of \cite{mittalstealthy} report that a successful attack
requires at least a few minutes of traffic on Tor. Similarly,
ephemeral circuits provide an effective protection against known
attacks such as the predecessor attack \cite{wright2004predecessor}.

\smallskip
\noindent {\bf Consumers, Producers and ARs. } Each consumer runs
several processes that generate interests. For our analysis,
interests are created by a specific interface of a host, and the
corresponding content is delivered back to the same interface.
Interest encryption is either performed on the consumer's host, or on
an entity that routes consumer's traffic. In the latter case, the
channel between the user and the anonymizing entity is considered
secure.

Content is generated by producers, i.e., entities that can sign data.
We do not assume the correspondence between a producer and a
particular host. Content can be either stored in routers' caches, at
servers or dynamically generated in response to an interest.

ARs perform interests decryption and content encapsulation. Each AR
advertises a public key for signature verification and one or more
public keys for encryption. ARs must refresh their encryption keys
frequently, discarding old keys after a short grace period. In order
to simplify key distribution and allow consumer to immediately trust
new public keys from routers, we use a simple key hierarchy where a
long lived public key owned by the router (the signing key), is used
to certify short lived encryption keys. The signing key may be
certified by other entities using techniques like web-of-trust or
PKI.

\smallskip
\noindent {\bf Denial-of-service Attacks. } \ANON{} is envisioned as
a public overlay network and is clearly susceptible to DoS attacks.
Since anyone can join \ANON{} as an AR or use it as a consumer, we
make no distinction between insider and outsider attacks. The
adversary can send numerous interests to ARs or construct ephemeral
circuits longer than two hops in order to maximize effectiveness of
attacks. Moreover, it can consume AR resources by sending malformed
encrypted interests that require ARs to perform expensive and
ultimately useless public key decryption. Similar to Tor,
before establishing an ephemeral circuit, an AR can ask a consumer to
solve an easy-to-verify/expensive-to-solve puzzle. This and similar
techniques for \ANON{} are subjects of future work.
In a setting with long-lived circuits, such as Tor, disrupting a
node effectively shuts down all circuits that include it. Due to the
short lifespan of our ephemeral circuits, the same attack on \ANON{} only
causes a very small number of interests/data packets per user to be  dropped.

\smallskip
\noindent {\bf Abuse. } Similar to any other anonymity service,
\ANON{} can be  abused for a variety of nefarious purposes. We do not
elaborate on this topic. However, exit policies similar to those in
Tor \cite{tor} can be used with \ANON{} based on content names.

\newcommand{\adv}{{Adv}}
\newcommand{\ent}{d}

\section{Security Analysis}
\label{sec:security_analysis}
In this section we propose a formal model for evaluating the security
of \ANON{}. We define consumer anonymity and unlinkability with
respect to an adversary within this model. We finally provide
necessary and sufficient conditions for anonymity and unlinkability.
As our analysis shows, we are able to obtain a level of anonymity
comparable to Tor with two --- rather than Tor's three --- ARs thanks to
the lack of source addresses in NDN interests.

In general, efficacy of \ANON{} depends on the inability of the
adversary to correlate input and output of a non-compromised AR, and
its inability to observe all producer and consumers at the same time.
Since \ANON{} is designed for low-latency traffic, we do not
intentionally delay messages or introduce dummy packets, other than
some limited padding. This is similar to how Tor and other
low-latency anonymizing networks forward traffic, and implies that
traffic patterns remain almost unchanged as they pass through the
network~\cite{MurdochW08}. It is well known that, in Tor, this allows
the adversary that observes both ends of a communication flow to
confirm a suspected link between them~\cite{bauer2007low,
overlier2006locating}. For this reason, a {\em global passive
adversary} can violate anonymity properties of both Tor and \ANON{}.
However, we believe that such an adversary is unrealistic in 
a geographically distributed network spanning over multiple administrative domains, and designing against it would result in overkill.

We assume that any adversary monitoring all interfaces of an AR can
correlate entering encrypted traffic with its exiting, decrypted
counterpart using timing information. 
However, we believe that the short lifespan of ephemeral circuits -- and therefore
the limited number of related packets traveling through a single AR -- 
severely limits the adversary's ability to carry out this
attack. Unfortunately, at the time of this writing we do not have
enough experimental evidence to confirm this. For the sake of safety,
in the analysis below we assume that, by compromising all interfaces
of an AR, the adversary also compromises the AR itself. Therefore, a non-compromised AR must have at least one non-compromised interface.
To sum up, we assume that:
\begin{assumption}\label{assumption2}
$\adv$ cannot correlate input and output of a non-compromised AR.
\end{assumption}

Our analysis is based on {\em
indistinguishable configurations}. A configuration defines consumers'
activity with respect to a particular network. $\adv$ only controls a
subset of network entities and observes only some packets. Therefore,
it cannot distinguish between two configurations that vary only in
the activity that it cannot directly observe or in the content of
encrypted packets that it cannot decrypt.
In order to provide meaningful anonymity guarantees, we identify a
set of configurations that have one or more equivalent counterparts.
However, unlike \cite{feigenbaum2007model}, our analysis takes into
account the infrastructure underlying \ANON{}, i.e., the network
topology and packets exchanged over the {\em actual} network.
We believe that this makes our model and analysis both realistic and
fine-grained, since it accounts for all adversarial advantages
related to the underlying network structure. Packets sent by a
non-compromised consumer $u$  to a non-compromised AR $r$ transit
through several --- possibly compromised --- NDN routers that are not
part of \ANON. The model of \cite{feigenbaum2007model} considers $r$ compromised
even if only one link between $u$ and $r$ is controlled by the
adversary. Whereas, in our model, $r$ is considered to be
non-compromised.

\medskip

\subsection*{Notation and Definitions}
Table \ref{tab:notation} summarizes our notation. The intersection of
${\sf P}$ and ${\sf C}$ might not be empty, which reflects the fact
that consumers can act as producers and {\em vice versa}. Similarly,
our model does not prevent routers  from being producers and/or
consumers. Therefore, ${\sf R} \cap {\sf P}$ and ${\sf R} \cap {\sf
C}$ might be non-empty.

\begin{table*}[t!]%
\begin{center} \footnotesize
\begin{tabular}{|r|l||r|l|}
\hline
${\sf C}$  & set of all consumers, $u \in {\sf C}$ & $\adv$\ & adversary
\\ \hline
${\sf P}$  & set of all content producers, $p \in {\sf P}$ & $\ent$\ &  an entity, i.e., a router or a host
\\ \hline
${\sf R}$ &  set of all routers, $r \in {\sf R}$ & $\ent\ \rightarrow_\INT r$ & entity $\ent$\ sends interest $\INT$\ to some interface of router $r$
\\ \hline
${\sf IF}$ &  set of all interfaces on all network devices & $\INT \leadsto {\sf if}^r_i$ &  router $r$ receives interest $\INT$ on interface ${\sf if}^r_i$
\\ \hline
${\sf if}^r_i \in {\sf IF}$   &	 $i$-th interface on router $r$ & $\mathcal{E}_{pk}(\cdot)$ & CCA-secure hybrid encryption scheme
\\ \hline
$\mathcal{PK}$ & set of all public keys & $\INT_{pk_1, pk_2}$ & interest encrypted under public keys $pk_1, pk_2$

\\ \hline
$(pk_i, sk_i)$  &  public/priv.~encryption keypair of an AR $r_i$ & $\perp$ & no encryption
\\ \hline
\end{tabular}
\end{center}
\vspace{-0.5cm}
\caption{Notation.} %
\label{tab:notation}
\end{table*}%

The adversary is defined as a 4-tuple: $\adv=({\sf P}_{\adv}, {\sf
C}_{\adv}, {\sf R}_{\adv}, {\sf IF}_{\adv}) \subset ({\sf P}, {\sf
C}, {\sf R}, {\sf IF})$ where individual components specify
(respectively) sets of: compromised producers, consumers, routers and
interfaces. If $r\in {\sf R}_{\adv}$, then $\adv$ controls all
interfaces and has access to all decryption key and state information
of $r$. If all interfaces of $r$ are in ${\sf IF}_{\adv}$, then $r\in
{\sf R}_{\adv}$. In other words, for the sake of this analysis,
controlling all interfaces of a router is equivalent to learning that
router's decryption/secret key. We emphasize that for $r \in {\sf R}$ to be
non-compromised, at least one of its interfaces must be non-compromised.
If $p \in {\sf P}_{\adv}$, $\adv$ controls $p$'s interfaces, monitors
interests received by $p$ and controls both content and timing of
$p$'s responses to incoming interests. If $c \in {\sf C}_{\adv}$,
then $\adv$ controls all fields and timing of interests. Finally, if
${\sf if} \in {\sf IF}_{\adv}$, then $\adv$ can listen to all traffic
flowing through ${\sf if}$, as well as sending new traffic from it.
${\sf IF}_\adv$ includes all the interfaces
of compromised consumers, producers and routers {\em plus} additional
interfaces eavesdropped on by $\adv$.

For ease of notation, we do not explicitly indicate the name of the
next router in interest packets nor symmetric keys chosen by
consumers. We denote encrypted interests as:
$$
\INT_{pk_1, pk_2} = \mathcal{E}_{pk_1}(\mathcal{E}_{pk_2}(\INT))
$$
with $ pk_1, pk_2 \in \mathcal{PK} \ \cup \{\perp\}$
where $\perp$ indicates a special symbol for ``no encryption''.
If $pk_1=\ \perp$ then
$pk_2=\ \perp$. The size of public keys is a function of the global
security parameter $\kappa$. For simplicity, we denote $\INT_{pk_1,
\perp}$ as $\INT_{pk_1}$. When an AR receives $\INT_{pk_1, pk_2}$ and
if it is in possession of the decryption key corresponding to $pk_1$,
it removes the outer layer of encryption. While $\mathcal{E}$ is
CCA-secure (and therefore also CPA-secure), we do not require
$\mathcal{E}$ to be key private~\cite{BellareBDP01}. Key privacy
prevents an observer from learning the public key used to generate a
ciphertext. In \ANON{}, knowledge of the public key  used to encrypt
the outer layer of an interest does not reveal any more information
than the (cleartext) name on the interest.

We define the anonymity set with respect to interface ${\sf if}_i^r$ as:
$$
A_{{\sf if}_i^r} = \{ \ent \; |  \; \Pr\left[\ent\rightarrow_\INT r  \ |\  \INT \leadsto {\sf if}^r_i\right] > 0  \}
$$
In other words, for each interface ${\sf if}_i^r$ of router $r$,
$A_{{\sf if}_i^r}$ contains all entities that could have sent  $\INT$
with non-zero probability. We define ${\sf path}^\INT = \{ {\sf
if}^r_i \; | \;  \INT \leadsto  {\sf if}^r_i \}$. This is the
sequence of interfaces traversed by $\INT$. We use it to define the
anonymity set of an interest with respect $\adv$:
\[
A^\INT_{{Adv}} \triangleq \bigcap_{{\sf path}^\INT \cap{\sf IF}_{Adv}} A_{{\sf if}^r_i}
\]
Intuitively, if $u$ is far away from a compromised entity $\ent$,
then all sets
$A^\INT_{{Adv}}$ such that $u \in A^\INT_{{Adv}}$ are a large subset
of ${\sf C}$. $\adv$\ can rule out possible senders of an interest
(i.e., determine if $u \notin A^\INT_{{Adv}}$) only if it controls at
least one entity (routers, interfaces) along each path that $u$ does
not share with other consumers. The level of anonymity of $u\in
A^\INT_{{Adv}}$ with respect to $\adv$ is proportional to the size of
$A^\INT_{{Adv}}$. In particular, if $u$ is the only member of
$A^\INT_{{Adv}}$, it has no anonymity, since $\INT$ must have been
issued by $u$.

A configuration is a description of the network activity. Each
configuration maps consumers to their actions, defined as the
interest they issue and the corresponding content producers. More
formally, a configuration is a relation:
$$
C:{\sf C} \rightarrow \{(r_1, r_2, p, \INT_{pk_1, pk_2})\}
$$
with $(r_1, r_2, p, \INT_{pk_1, pk_2})\in {\sf R}^2\times {\sf P} \times \{0,1\}^*$, that maps a consumer to: a pair of routers defining an ephemeral
circuit, an interest (encrypted for this circuit) and a producer.
$C(u)$ is a 4-tuple that represents one action of $u$ in $C$. $C_i$
is the selection on the  $i$-th component of $C$, i.e., if $C(u) =
(r_1, r_2, p, \INT_{pk_1, pk_2})$, then $C_1(u) = r_1$, $C_2(u) =
r_2$, $C_3(u) = p$ and $C_4(u) = \INT_{pk_1, pk_2}$.

We say that two configurations $C$ and $C'$ are ``indistinguishable
with respect to $\adv$'' if $\adv$ can only determine with
probability at most $1/2+\varepsilon$ which configuration
corresponds to the observed network, for some $\varepsilon$
negligible in the security parameter $\kappa$. We denote two such
 configurations as $C\equiv_{Adv} C'$.

We now show that assumption \ref{assumption2} holds if a passive
adversary observes only input and output values of an AR (i.e., it
cannot use timing information or other side-channels), and the
underlying encryption scheme is semantically secure. Claim
\ref{claim1} below states that, for any encrypted interest, $\adv$
cannot determine if it corresponds to an interest decrypted by a
non-compromised router, by observing the two and with no additional information.

\begin{claim}\label{claim1}
Given any CPA-secure public key encryption scheme $\mathcal{E}$ and
two same-length interests $\INT^0, \INT^1$ chosen by $\adv$, $\adv$
has only negligible advantage over $1/2$ in determining the value of
a randomly selected bit $b$, given $\INT^b_{pk_1, pk_2}$,
$\INT^0_{pk_2}$ and $\INT^1_{pk_2}$, with $pk_1 \in \mathcal{PK}$ and
$pk_2 \in \mathcal{PK} \cup \{\perp\}$.
\end{claim}

\noindent
Due to the lack of space, Claim \ref{claim1} is formally justified in
Appendix \ref{appendix_security_proofs}.

\subsection*{Anonymity Definitions and Conditions}

In this section we present formal definitions of anonymity for our
model. We introduce the notions of {\em consumer anonymity}, {\em
producer anonymity} and {\em producer and consumer unlinkability}. We
show that ephemeral circuits composed of two anonymizing routers ---
at least one of which is not compromised --- provide consumer {\em
and} producer anonymity. This, in turn, implies consumer and producer
unlinkability.
Due to the lack of space, we defer formal proofs of the theorems in
this section to Appendix \ref{appendix_security_proofs}.

A consumer $u$ enjoys {\em consumer anonymity} if $Adv$ cannot determine whether $u$ or a different user $u'$ is retrieving some specific content. This notion is formalized using indistinguishable configurations: given a configuration $C$ in which $u$ retrieves content $t$, $u$ has consumer anonymity if there exist another configuration $C'$ in which $u'$ retrieves $t$ and $\adv$ cannot determine whether he is observing $C$ or $C'$.
More formally:

\begin{definition}[Consumer anonymity]
$u \in ({\sf C} \setminus {\sf C}_\adv)$ has consumer anonymity in
configuration $C$ with respect to $\adv$\ if there exists
$C'\equiv_{\adv} C$ such that $C'(u') = C(u)$ and $u'\neq u$.
\end{definition}

\begin{theorem}\label{cons_anon}
$u \in ({\sf C} \setminus {\sf C}_\adv)$ has consumer anonymity in
$C$ with respect to ${Adv}$ if there exists $u'\neq u$ such that any of
the following conditions hold:
\begin{compactenum}
\item $u, u' \in A^{C_4(u)}_{{\adv}}$ 
\item $C_1(u) = C_1(u')$, $C_1(u) \notin {\sf R}_{\adv}$ and $C_1(u) \in
    A^{\INT_{pk_2}}_{{\adv}}$ where $C_4(u) = \INT_{pk_1, pk_2}$
    
\item $C_2(u) = C_2(u')$, $C_2(u) \notin {\sf R}_{\adv}$  and $C_2(u) \in
    A^\INT_{{\adv}}$ where $C_4(u) = \INT_{pk_1, pk_2}$
    
\end{compactenum}
\end{theorem}

\noindent
Informally, the theorem above states that \ANON{} provides consumer anonymity with respect to $\adv$ if: {\em 1.} $\adv$\ cannot observe encrypted interests coming from $u$ and $u'$, or it cannot distinguish between the two consumers due to anonymity provided by the network layer; or {\em 2.} $u, u'$ share an non-compromised first router in at least
one ephemeral circuit; or {\em 3.} $u,
u'$ share an non-compromised second router  in at least one
ephemeral circuit.

Similarly to consumer anonymity, producer anonymity is defined in terms of indistinguishable configurations. In particular, a producer $p$ enjoys anonymity with respect to $\adv$ which observes $\INT_{pk_1, pk_2}$ if $\adv$ cannot distinguish between a configuration $C$ where $p$ produces the content corresponding to $\INT$  and a configuration $C'$ where $p'$ and not $p$ produces {\em that} content.

\begin{definition}[Producer anonymity]
Given $\INT_{pk_1, pk_2}$ for $p\in{\sf P}$, $u \in {\sf C}$ has
producer anonymity in configuration $C$ with respect to $p, \adv$\ if
there exists an indistinguishable configuration $C'$ such that
$\INT_{pk_1, pk_2}$ is sent by a non-compromised consumer to a
producer different from $p$.
\end{definition}

\begin{theorem} \label{prod_anon}
$u$ has producer anonymity in $C$ with respect to $p, \adv$\ if any
of the following conditions hold:
\begin{compactenum}
\item There exists $C(u)$ such that $C_1(u)$ (the first anonymizing
    router) is not compromised and $C_4(u) = \INT_{pk_1, pk_2}$,
    $C_1(u) = C_1(u')$ and $C_3(u) = p \neq C_3(u')$ for some
    non-compromised $u'\in {\sf C}$, {\em or}
\item There exists $C(u)$ such that $C_2(u)$ (the second anonymizing
    router) is not compromised and $C_4(u) = \INT_{pk_1, pk_2}$,
    $C_2(u) = C_2(u')$ and $C_3(u) = p\neq C_3(u')$ for some
    non-compromised $u'\in {\sf C}$
\end{compactenum}
\end{theorem}

\noindent
Finally, we define producer and consumer unlinkability as: 
\begin{definition}[Producer and consumer unlinkability]
We say that $u \in ({\sf C} \setminus {\sf C}_\adv)$ and $p \in {\sf
P}$ are unlinkable in $C$ with respect to $\adv$\ if there exists $C'
\equiv_\adv C$ where $u$'s interests are sent to a producer $p'\neq
p$.
\end{definition}

\begin{corollary}
Consumer $u \in ({\sf C} \setminus {\sf C}_\adv)$ and producer $p \in
{\sf P}$ are unlinkable in configuration $C$ with respect to $\adv$\
if $p$ has producer anonymity with respect to $u$'s interests {\em
or} $u$ has consumer anonymity and there exists a configuration
$C'\equiv_\adv C$ where $C'(u')=C(u)$ with $u'\neq u$ and $u'$'s
interests have a destination different from $p$.
\end{corollary}

\begin{corollary}
Consumer $u \in ({\sf C} \setminus {\sf C}_\adv)$ and producer $p \in
{\sf P}$ are unlinkable in configuration $C$ with respect to $\adv$\
if both producer and consumer anonymity hold.
\end{corollary}

We emphasize that this result also holds for ephemeral circuits with
length greater than two ARs.%

\section{Implementation and Performance}
\label{sec:implem}
\ANON{} is implemented as an application-level service consisting of
client ``stack'' (used by consumers) and server
program that runs on \ANON{} ARs. Both are written in
C and interface to NDN through Unix domain
sockets.\footnote{At the time of this writing, there is no direct function interface to NDN}
Cryptographic algorithms are implemented using OpenSSL \cite{openssl}. Hybrid encryption is obtained using
RSA-OAEP \cite{RSA-OAEP} and AES+HMAC \cite{AES, HMAC}. The latter is also used for symmetric
encryption. We use SHA-256 for HMAC and 1024- and 128-bit keys for RSA and AES, respectively.
Loose time synchronization among \ANON{} client and servers are achieved using \texttt{\small pool.ntp.org}, a public pool of
NTP servers.

\ANON{} client encrypts interests from user applications. In order
to hide all possible sources of de-anonymizing information,
encryption is performed over the full interest packet,
including: name, scope, exclusion filters and duplicate
suppression string fields.
Following NDN ``rules'', \ANON{} AR announces the ability to serve the
root (``{\small \tt /}'') namespace and receives all traffic sent from (or to) the
local NDN routing process. This allows traffic to be routed
through \ANON{} by default, requiring no changes to existing applications.
For more granularity, consumers can vary the default namespace, e.g., ``{\tt \small /andana/}''.
However, this would require privacy-seeking applications to explicitly direct their traffic to that
namespace, similar to today's configurable proxy settings.

\ANON{} servers run as applications on NDN routers. Each server is
responsible for its {\em relay} and {\em session creation}
namespaces. The former is a globally routable namespace  used for
receiving both session-based and asymmetrically encrypted Interests.
Clients using session-based encryption in \ANON{} need to first establish symmetric keys with
servers. To start a new session with a server, a clients sends an interest in the {\tt \small
createsession} namespace, registered by the server code as a sub-prefix of the
relay namespace.
We deployed our prototype and run a series of tests on the Open Network Laboratory
(ONL) \cite{ONL}. ONL is a testbed developed by Washington University to enable
experimental evaluation of advanced networking concepts in a realistic environment.
To guarantee highly reproducible results, ONL provides reservation-based exclusive
access to most of its host and network resources. All our experiments used single-core Linux
machines with 512 MB of RAM and gigabit switches (one machine per switch).

We compare plain NDN and \ANON{} on a simple line topology with four switches and four Linux machines, each corresponding to an NDN node.
Static routing is established between nodes. The first NDN node in the line topology
acts as a consumer and runs {\small \tt ccngetfile} --- a small tool from CCNx open-source
library that retrieves data published as NDN content and stores it in a local file.
We performed tests with 1, 10, and 100MB files; each file was retrieved from the NDN
repository of the machine at the other end of the line topology. Results of this comparison
for 10MB files are summarized in Fig.~\ref{fig:andora}. Due to space constraints,
we illustrate all file retrieval results in Appendix \ref{appendix_performance}.
Results show that computational overhead introduced by \ANON{} roughly doubles download
times over plain NDN. This is assuming an almost-perfect world where ARs topologically
align with the best path and link bandwidths are abundant.
\begin{figure}[tb]
\centering
\includegraphics[width=.49\textwidth]{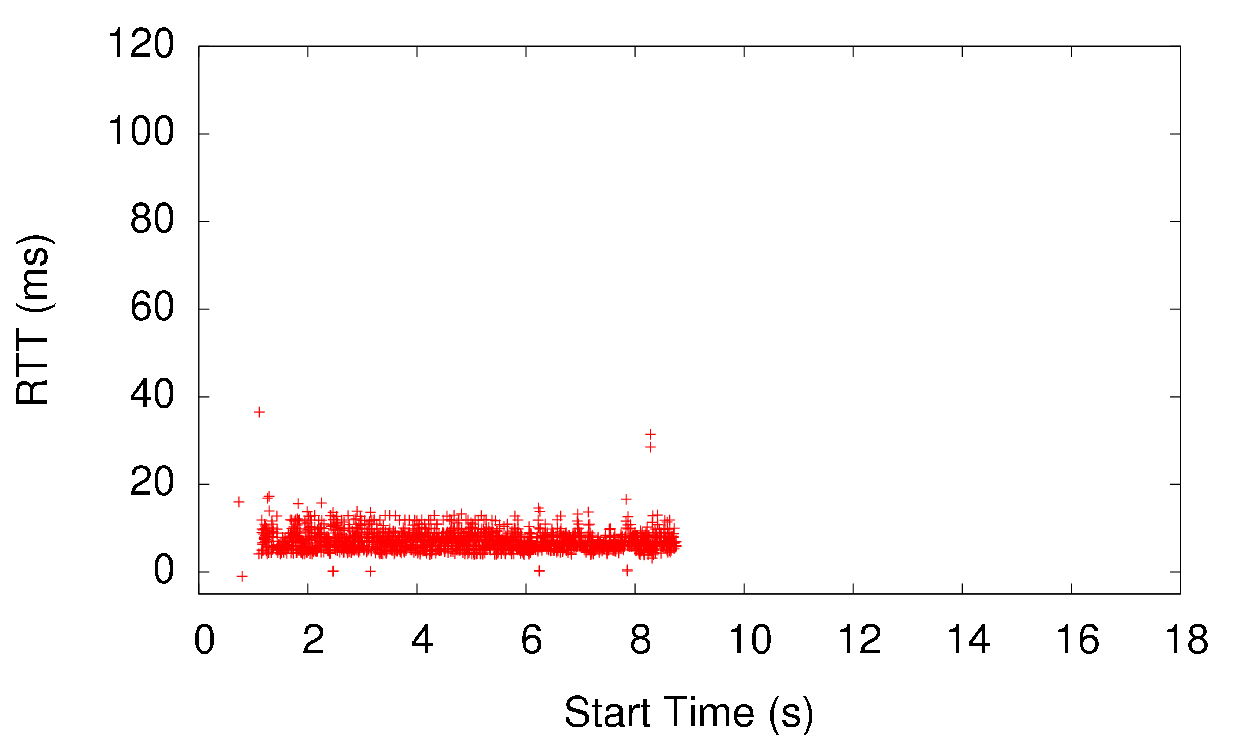}
\includegraphics[width=.49\textwidth]{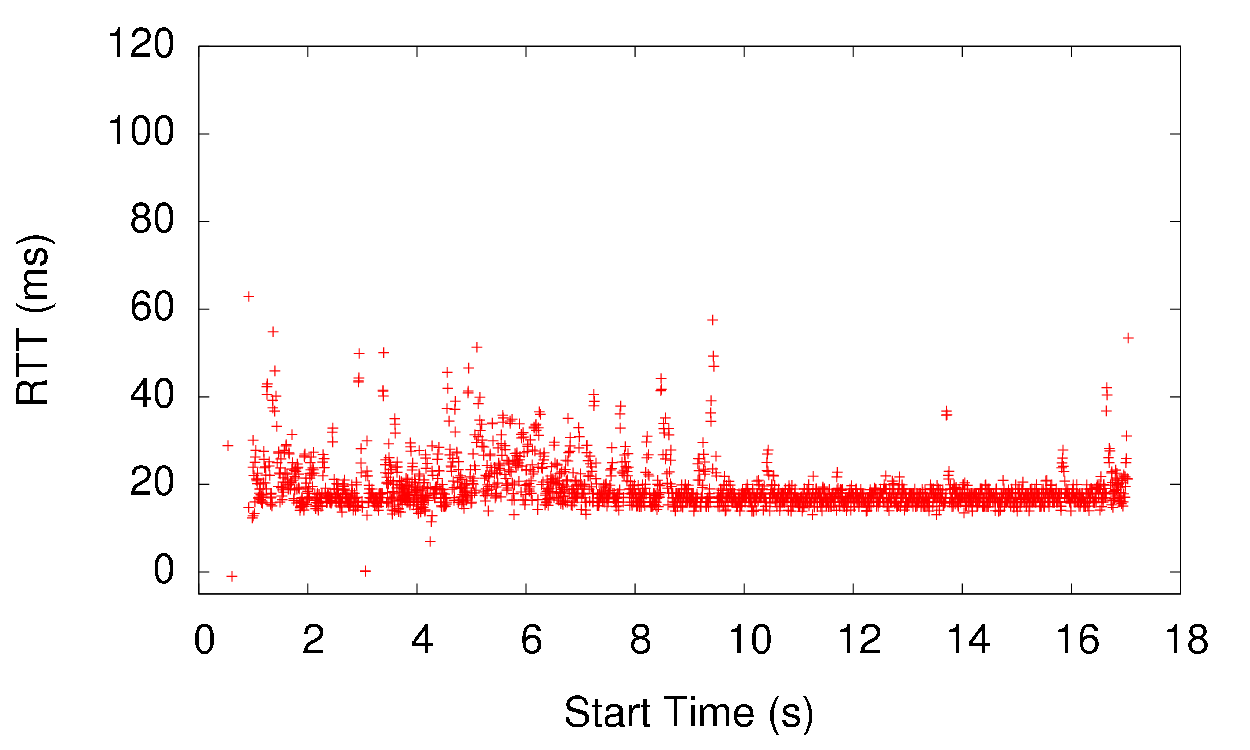}
\vspace{-0.2cm}
\caption{\small Left: RTT for 10MB of content over NDN (limited anonymity).
Right: RTT for 10MB of content over \ANON{} (full anonymity).}
\label{fig:andora}
\end{figure}

In order to compare \ANON{}'s computational overhead with a similar anonymizing tool,
we deployed Tor over ONL and measured its overhead over TCP/IP.
We measured performance of TCP/IP baseline deploying five switches, connected in a line,
and two Linux machines (one at each end): the first acting as client (running {\small \tt curl}),
the second -- as server (running {\small \tt lighttpd} HTTP server).
Performance of Tor was measured on a topology that closely mimics that of TCP/IP baseline: five switches,
connecting three Tor relays, a client and a server. To ensure ``line'' topology, Tor client is configured to use
explicit entry and exit nodes; DNS lookups are avoided by using IP addresses in all tests.

Before discussing the results, we mention some comparison details. NDN is a research project and its
code is optimized for functionality rather than performance. It provides content authentication through
digital signatures -- a computationally expensive feature not present in either TCP/IP or Tor.
NDN stack currently runs as a user-space application, in contrast to TCP/IP that runs in kernel-space.
Finally, in all our experiments, NDN had to run on top of TCP/IP (rather than at layer 2) due to limitations
of the underlying ONL testbed. Consequently, we believe a fair comparison between \ANON{}
and Tor can only be achieved by focusing the analysis on \emph{relative} overhead imposed
by each, over the network it is deployed, i.e., NDN and TCP/IP respectively.

Figure~\ref{fig:download} shows the performance of \ANON{} and Tor with respect to their baselines.
The graph on the left shows the measurements including the time required to setup a Tor circuit
and all ephemeral circuits for \ANON{}. Session-based \ANON{} is denoted by \ANON{}-S,
while \ANON{} with asymmetric encryption is referred to as \ANON{}-A. For small- to medium-size
files (1-10MB), overhead of \ANON{}-A is between 1.5$\times$ and 1.75$\times$. As expected,
\ANON{}-S exhibits lower overhead (1.45$\times$ to 1.7$\times$) due to more efficient symmetric encryption.

In comparison, Tor's download time for the same amount of data is between 2.3 and 7 times higher
than that of TCP/IP. This imposes significant overhead for content size that fits many typical web pages.
Whereas, \ANON{} is efficient in anonymizing such traffic patterns.
Large file transfers are more efficient with Tor, which increases the total download time by about
1.4 times, compared to 2.4 and 2.1 of \ANON{}-A and \ANON{}-S.

The right-side graph in Figure~\ref{fig:download} shows the relative speed of three approaches without
including circuit setup time. Our measurements show that overhead of ephemeral circuit creation in
\ANON{}-S is negligible. Since a new ephemeral circuit must be selected for every interest with
\ANON{}-A, we simply report the same values from the previous graph. Results confirm that
overhead of circuit creation in Tor is significant when retrieving small-size content. Removing
this initialization phase from the measurements significantly reduces Tor's overhead.
However, the overhead of \ANON{} with respect to its baseline is still smaller than that
of Tor for content up to 10MB.

In absolute terms (comparing raw download times), Tor + TCP/IP
performs better than \ANON{} + NDN in our testbed experiments.
However, we believe that, in a realistic geographically-distributed
deployment setting with limited-bandwidth links, \ANON{} + NDN would
provide a significant performance advantage over Tor + TCP/IP due to
its shorter (ephemeral) circuits. In other words, we anticipate that
shorter circuits and content caching in \ANON{} + NDN would result in
appreciably lower overall download times than Tor + TCP/IP in a
global internet setting.
\begin{figure}[tb]
\centering
\includegraphics[width=.49\textwidth]{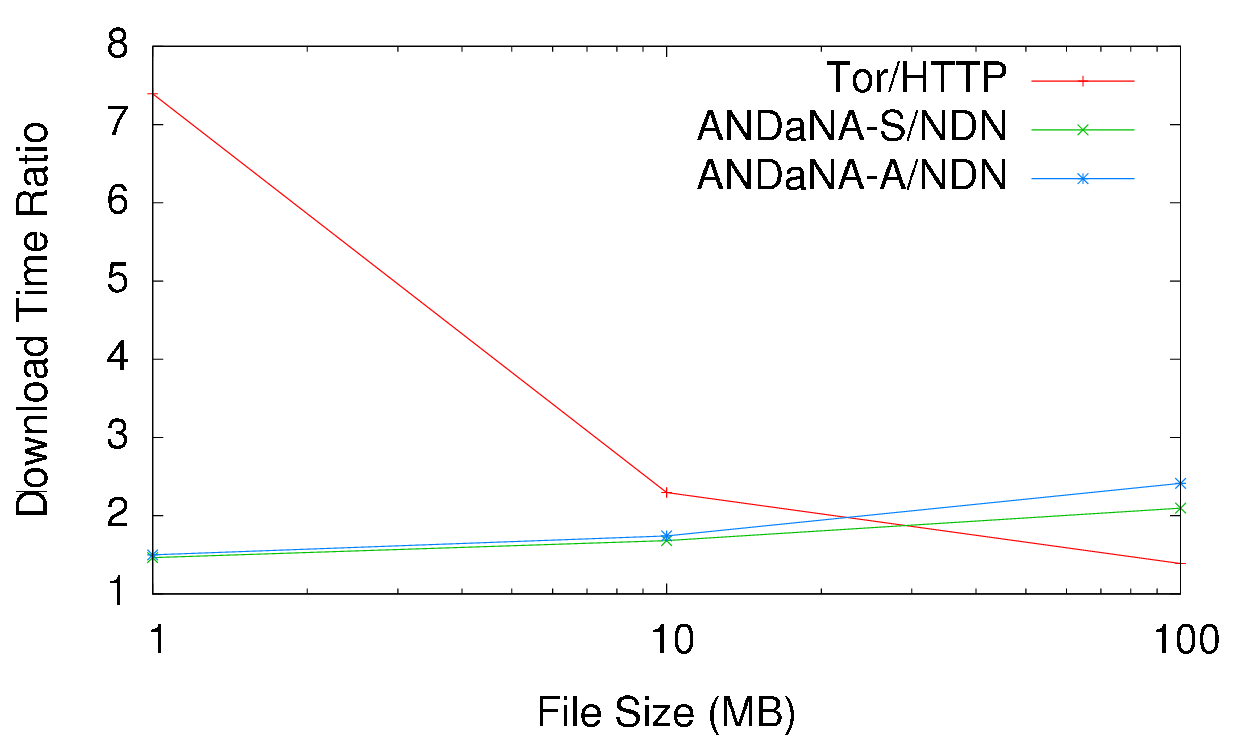}
\includegraphics[width=.49\textwidth]{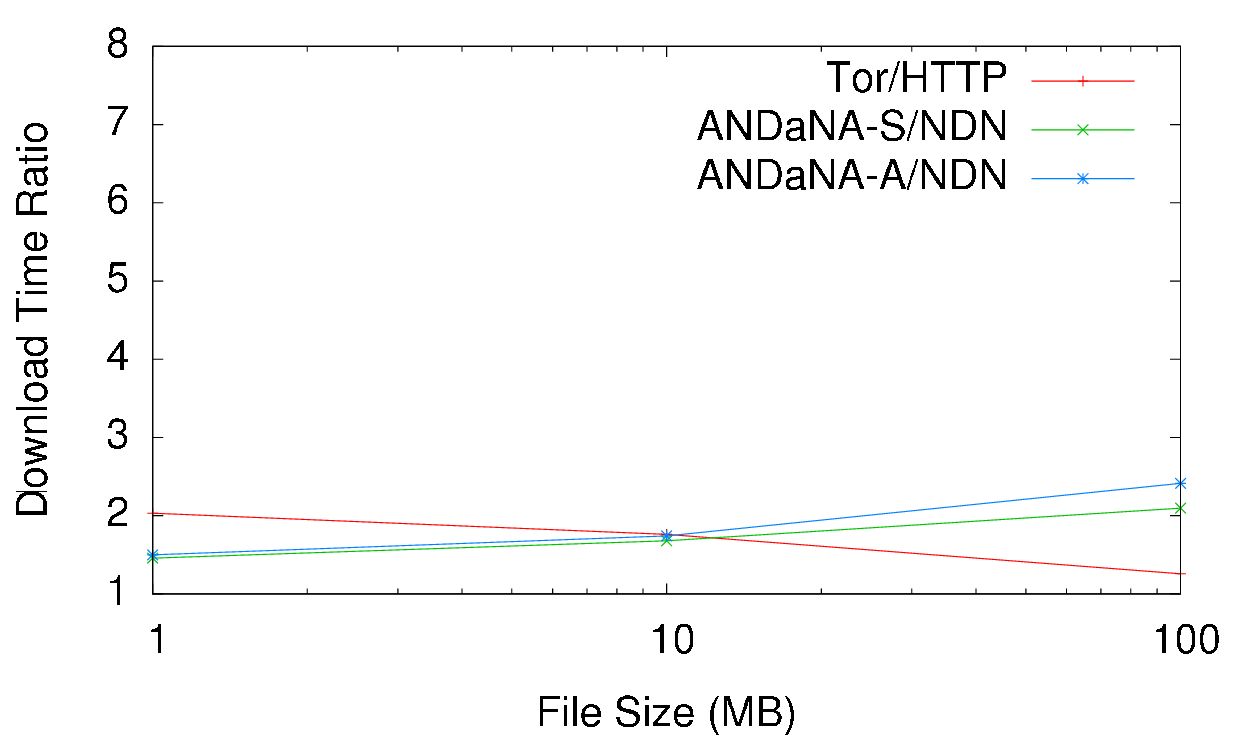}

\caption{\small Comparison of 1, 10, and 100MB file download times over Tor, \ANON{}-S and \ANON{}-A with
respect to respective baselines. Left: transfer time {\em and} circuit setup time. Right: transfer time only.}
\label{fig:download}
\end{figure}

\section{Conclusions and Future Work}
\label{sec:conclusions}
Content-centric networking is a major transition from today's
world that focuses on communication end-points. NDN project
represents one of the most visible current research efforts aiming to bring
content-centric networking into the foreground by using it as a
possible future Internet architecture.
Despite some privacy-friendly features and side-effects, NDN poses
some interesting privacy challenges. This work presents an initial
attempt to provide anonymity in NDN. The main contribution of
this work is threefold: (1) exploration of privacy issues in
NDN, (2) design of an anonymization tool -- \ANON, and (3) its
security analysis and performance assessment.

At the same time, particularly because the entire NDN project (and,
of course, \ANON) represent work-in-progress, one of the main goals
of this paper is to solicit comments from the security research
community. Also, since our work merely scratches the surface of
privacy issues in content-centric networking and NDN, a number of
issues are left for future work, including:
\begin{itemize}
\item More performance experimentation with \ANON, especially, in
    larger testbeds and under various traffic load / congestion
    scenarios. (This should lead to better code profiling and lower
    overhead.)
\item Comprehensive directory service for effective large-scale
    distribution of up-to-date AR information.
\item In-depth study of both privacy and performance trade-offs in
    the use of asymmetric vs.~symmetric \ANON{} variants.
\item DoS mitigation measures, such as computational puzzles for
    circuit establishment.
\item Red-teaming experiments to assess realistic privacy attainable
    with \ANON.
\item Modification of \ANON\ to support other emerging content-centric
    architectures and comparative experiments among them.
\end{itemize}

\section*{Acknowledgments}
We thank NDSS'12 anonymous reviewers for their valuable feedback. We
are also grateful to Van Jacobson, Jim Thornton, Kasper Rasmussen,
Yanbin Lu, Lixia Zhang and Mark Baugher for their helpful input and
comments on earlier drafts of this paper. This work was conducted in
the context of the NSF project: ``CNS-1040802: FIA: Collaborative
Research: Named Data Networking (NDN)''.

\bibliographystyle{abbrv}
{
\bibliography{andora}
}
\appendix

\section{Security Proofs}
\label{appendix_security_proofs}

\noindent {\em Justification of Claim \ref{claim1}:} Suppose that
Claim \ref{claim1} is false. Then, $\adv$\  can be used to construct
an algorithm ${\sf Sim}$ that breaks the CPA-secure encryption scheme
$\mathcal{E}$ as follows: ${\sf Sim}$ plays the CPA-security game
with a challenger, that selects a public key $pk$. ${\sf Sim}$
selects a public key $pk_2$ and initializes $\adv$, that eventually
returns two interests $int^0, int^1$ of its choice. ${\sf Sim}$ sends
$c_0 = \mathcal{E}_{pk_2}(\INT^0)$ and $c_1 =
\mathcal{E}_{pk_2}(\INT^1)$ to the challenger, that returns $c^* =
\mathcal{E}_{pk}(c_b) =
\mathcal{E}_{pk}(\mathcal{E}_{pk_2}(\INT^b))$. ${\sf Sim}$ sends
$(c^*, c_0, c_1)$ to the challenger that eventually returns its
choice $b'$. ${\sf Sim}$ outputs $b'$ as its choice. The output of
${\sf Sim}$ is $b'=b$ iff $\adv$ guesses $b'$ correctly. Since $\adv$
guesses $b'$ correctly with non negligible advantage over $1/2$,
 ${\sf Sim}$ breaks the CPA-security of $\mathcal{E}$ with non negligible advantage.
This violates the hypothesis of Claim \ref{claim1}, and, therefore, such $\adv$ cannot exist. \qed

\smallskip

\begin{proof}[Proof of Theorem \ref{cons_anon}  --- Consumer Anonymity (sketch)]
We prove that each condition in Theorem \ref{cons_anon} implies consumer anonymity:
\begin{compactenum}
\item Assume that, for each $u'\neq u$ there exists no configuration
$C'\equiv_{\adv} C$ with respect to $\adv$\ such that $C'(u') = C(u)$. $\adv$ cannot determine that
$C(u)\notin C'$ using only $C_2(u)$, $C_3(u)$ and $C_4(u)$: if $C_1(u) = C'_1(u')$ for some $C'\equiv_{Adv} C$
and $u'$ (i.e. there exist an indistinguishable configuration with respect to ${Adv}$ where a consumer different from
$u$ sends an interest to $C_1(u)$ through interface ${\sf if}^{C_1(u)}_i$ and $u, u' \in A_{{\sf if}^{C_1(u)}_i}$),
then there must exist a tuple $C'(u') = C(u)$ since (a possibly compromised) $r$ cannot process interests coming
from consumers in the same anonymity set differently -- that would imply that they are not in the same anonymity set.
Therefore, for each configuration $C'\equiv_{Adv} C$, and for each
$u'\neq u$ $\exists C'_1(u') = C_1(u) \Rightarrow \exists C'(u') = C(u)$.

For this reason,  $C'_1(u') \neq C_1(u)$ for all $C'\equiv_{\adv} C$ and for all $u'\neq u$, i.e. $\forall
C'_1(u') = C_1(u).C(u) \notin C'$. This is true if and only if $\adv$\ controls at least one interface
${\sf if}^r_i \in {\sf path}^{C_4(u)}$ for which $u'$ is not in the anonymity set of ${\sf if}^r_i$,
i.e., $\exists {\sf if}^r_i \in {\sf path}^{C_4(u)} \cap {\sf IF}_{\adv}$ s.t. $u' \notin A_{{\sf if}^r_i}$
Since this contradicts the hypothesis, there must exist a configuration $C'$ indistinguishable
from $C$ with respect to $\adv$\ such that $C'(u') = C(u)$.

\item We assume that, for each $u' \neq u$, $\adv$\ can distinguish between interests from
$u$ from those from $u'$ (i.e., condition 1 of theorem \ref{cons_anon} does not hold). We show
how to prove theorem \ref{cons_anon} by reduction. Assume that there exists an efficient adversary $\adv$\ such that ${\sf C}_\adv = {\sf C} \setminus \{u, u'\}$ and ${\sf R}_\adv = {\sf R} \setminus \{r_1\}$ (i.e., $\adv$ compromised all entities, except $u, u'$ and $r_1$). Suppose that
$C(u) = (r_1, r_2, p, \INT^0_{pk_1, pk_2})$, $C(u') = (r_1, r'_2, p', \INT^1_{pk_1, pk'_2})$ for some $r_2, r'_2, p, p', \INT^0, \INT^1$.
For each $C'$, $\adv$\ outputs: 1 on input of $C$ and 0 on input of $C'$ with non-negligible probability, where $C'(u) = C(u')$ and $C'(u') = C(u)$.
In other words, there is no configuration for which $C\equiv_{\adv} C'$ holds. We sketch how $\adv$\ can be used as a
subroutine in a simulator ${\sf Sim}$ that breaks Claim \ref{claim1}.

${\sf Sim}$ creates a random network topology $N$ and inputs it to $\adv$. ${\sf Sim}$ also inputs the information that $\adv$\ would obtain by compromising all entities in $N$ except $u, u'$ and $r_1$.
As such, ${\sf Sim}$ also includes $\INT^b_{pk_1,pk_2}$ and $\INT^0_{pk_2}$, $\INT^1_{pk_2}$ received from the challenger of Claim \ref{claim1} to the input of $\adv$.
Then, ${\sf Sim}$ sends to $\adv$ configurations
$C$ and $C'$, where $C$ is identical to $C'$, except that $C(u) = C'(u')$ and $C(u') = C'(u)$, and $C(u) \neq C(u')$.
We have that $b=1$ iff $\adv$\ outputs 1.
Since existence of ${\sf Sim}$ violates Claim \ref{claim1}, $\adv$\ cannot exits.

\item We assume that, for each $u' \neq u$, $\adv$\ can distinguish between interests from
 $u$ from those from $u'$ (i.e., condition 1 of theorem \ref{cons_anon} does not hold) and that the first router in
 $u$'s and $u'$'s paths is compromised, i.e., condition 2 of theorem \ref{cons_anon} does not hold. We then prove
 theorem \ref{cons_anon} by reduction. Assume that there exists an efficient adversary $\adv$\ such that ${\sf C}_\adv = {\sf C} \setminus \{u, u'\}$ and ${\sf R}_\adv = {\sf R} \setminus \{r_2\}$ (i.e., $\adv$ compromised all entities, except $u, u'$ and $r_2$). Suppose that
 $C(u) = (r_1, r_2, p, \INT^0_{pk_1, pk_2})$, $C(u') = (r'_1, r_2, p', \INT^1_{pk'_1, pk_2})$ for some $r_1, r'_1, p, p', \INT^0, \INT^1$.
For each $C'$, $\adv$\ outputs 1 on input of $C$, and 0 on input of $C'$, where $C'(u) = C(u')$ and $C'(u') = C(u)$.
In other words, there is no configuration where $C\equiv_{Adv} C'$ holds. We sketch how $\adv$\ can be used as a subroutine
in a simulator ${\sf Sim}$ to determine, given $\INT_{pk_2}$ and $\INT'_{pk_2}$, whether $\INT = \INT'$.

${\sf Sim}$ creates a random network topology $N$ and inputs it to $\adv$. ${\sf Sim}$ also inputs the information that $\adv$\ would obtain by compromising all entities in $N$ except for $u, u'$ and $r_2$.
${\sf Sim}$ interacts with the challenger of Claim \ref{claim1} setting the innermost key of its challenge, denoted as $\overline{pk_2}$, to $\perp$. ${\sf Sim}$ receives $\INT^b_{pk_1}$ for some $\INT^0, \INT^1$ of its choice, and adds $\INT^b_{pk_1, \overline{pk_2}}$, $\INT^b_{\overline{pk_2}}$ and $\INT^b_{\overline{pk_2}}$
 to the input of $\adv$.
Then ${\sf Sim}$ sends to $\adv$ configurations $C$ and $C'$, where $C$ is identical to $C'$ except that $C(u) = C'(u')$ and
$C(u') = C'(u)$, and $C(u) \neq C(u')$. We have that $b=1$ iff ${Adv}$ outputs 1.
Since the existence of ${\sf Sim}$ would violate Claim \ref{claim1}, $\adv$\ cannot exits. \qedhere
\end{compactenum}
\end{proof}

\bigskip

\begin{proof}[Proof of Theorem \ref{prod_anon} --- Producer Anonymity (sketch)]
We prove that each condition in Theorem \ref{prod_anon} implies producer anonymity:
\begin{compactenum}
\item Let $C_4(u') = \INT_{pk_1, pk'_2}'$ and let $C'$ be identical to $C$ except that
$C'(u) = (C_1(u), C_2(u), C_3(u), C_4(u'))$ and  $C'(u') = (C_1(u'), C_2(u'), C_3(u'), C_4(u))$.
In other words, $C'$ is a configuration where $\INT_{pk_1, pk_2}$ is sent to a producer different from $p$.
In this setting, $\adv$\ can only distinguish $C'$ and $C$ by distinguishing $C'(u)$ and $C'(u')$.
Claim \ref{claim1} guarantees that $\adv$\ that observes $\INT_{pk_1, pk_2}$ and $\INT'_{pk_1, pk'_2}$
cannot determine which corresponds to $\INT$ and which -- to $\INT'$. Moreover, Assumption
\ref{assumption2} prevents $\adv$\ from linking the output of non-compromised router $C_1(u)$ with
$\INT_{pk_1, pk_2}$ and $\INT_{pk_1, pk'_2}'$. Therefore, $C \equiv_{\adv} C'$.

\item Similarly, let $C_4(u') = \INT_{pk_1, pk'_2}'$ and let $C'$ be identical to $C$ except that
$C'(u) = (C_1(u), C_2(u),$ $C_3(u), C_4(u'))$ and  $C'(u') = (C_1(u'), C_2(u'), C_3(u'), C_4(u))$.
We assume that $C_1(u)$ and $C_1(u')$ are compromised. In this setting, $\adv$\ can only distinguish
between $C'$ and $C$ by distinguishing $C'(u)$ and $C'(u')$. Claim \ref{claim1} guarantees that any $\adv$\
that observes $\INT_{pk_1, pk_2}$ and $\INT'_{pk_1, pk'_2}$ cannot determine which corresponds to $\INT$
and which -- to $\INT'$. Moreover, Assumption \ref{assumption2} prevents $\adv$\ from linking the output of
non-compromised router $C_2(u)$ with $\INT_{pk_2}$ and $\INT_{pk'_2}'$. Therefore,
$C \equiv_{Adv} C'$. \qedhere

\end{compactenum}
\end{proof}

\onecolumn

\newpage

\section{Performance Evaluation: Additional Results}
\label{appendix_performance}

\begin{figure}[h!]
\centering
\includegraphics[width=.70\textwidth]{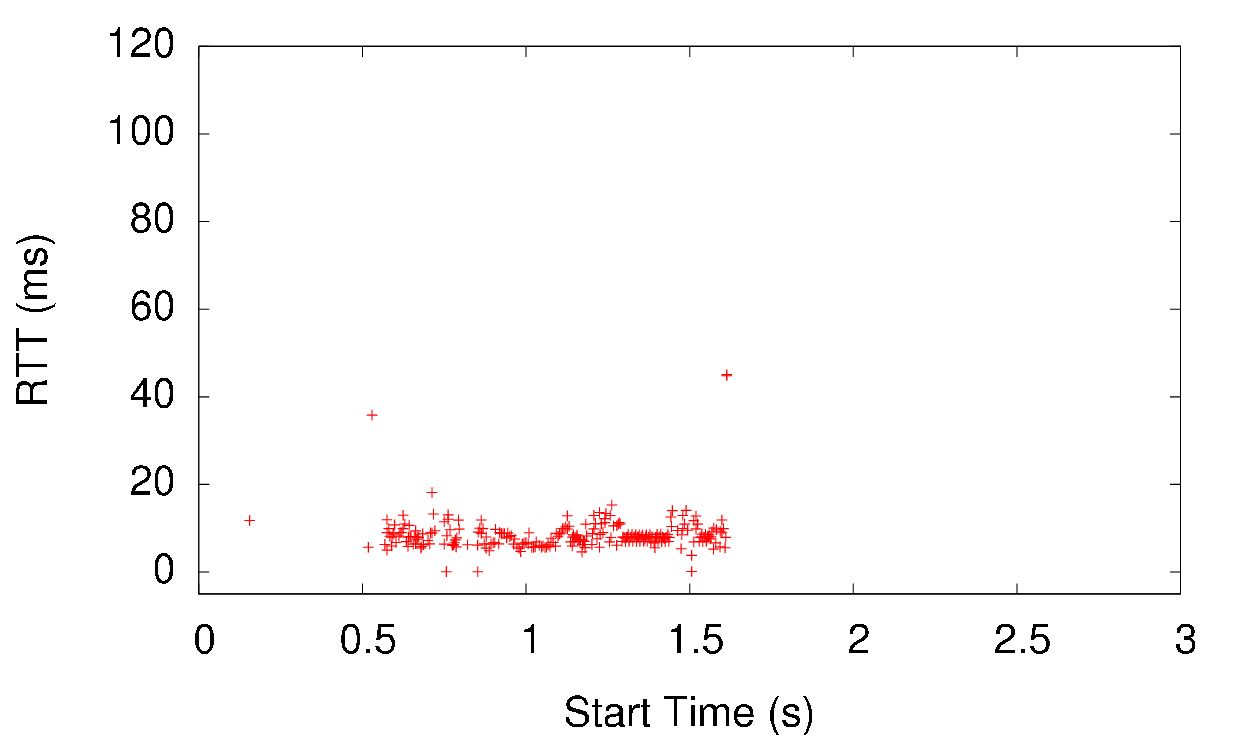} \includegraphics[width=.70\textwidth]{rtt_over_time_graphs/ccn_ccngetfile_rtt_over_time_test10MB}
\includegraphics[width=.70\textwidth]{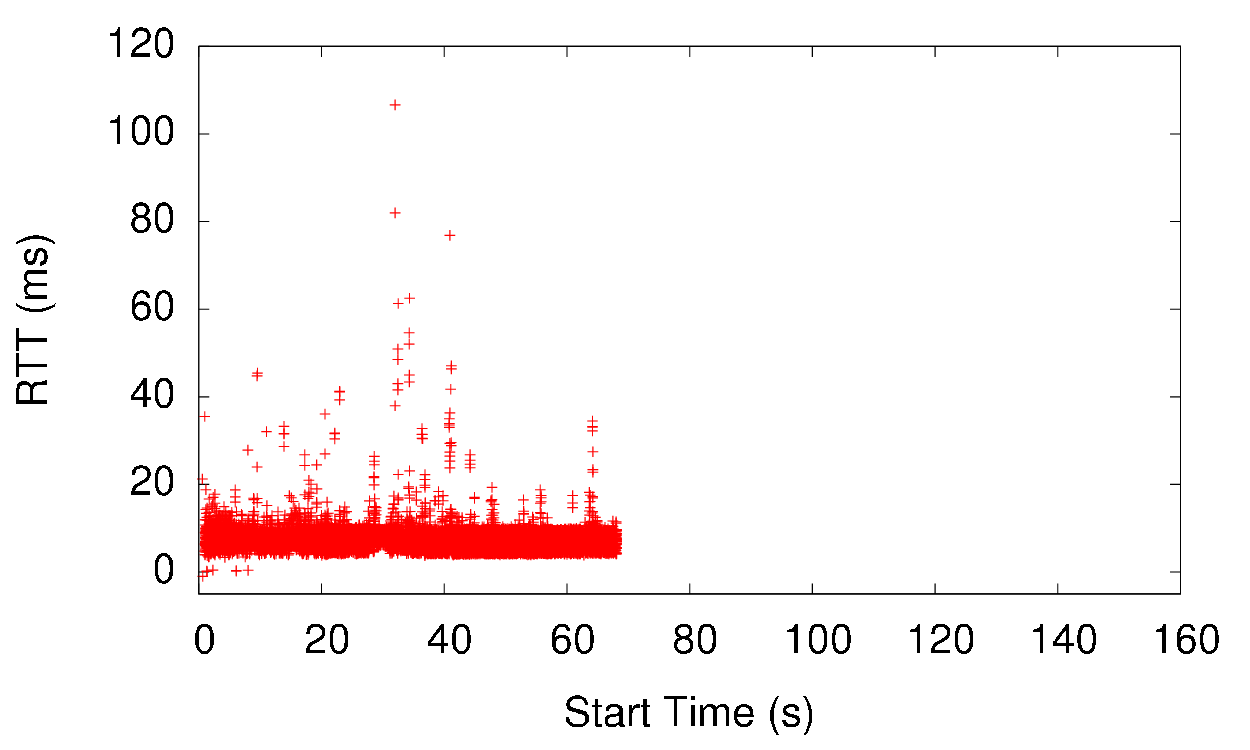}

\caption{\small Round trip time for transferring 1, 10 and 100MB of content over NDN (limited anonymity)}
\label{fig:ndn-all}
\end{figure}

\label{appendix_performance}
\begin{figure*}[htbp]
\centering
\includegraphics[width=.70\textwidth]{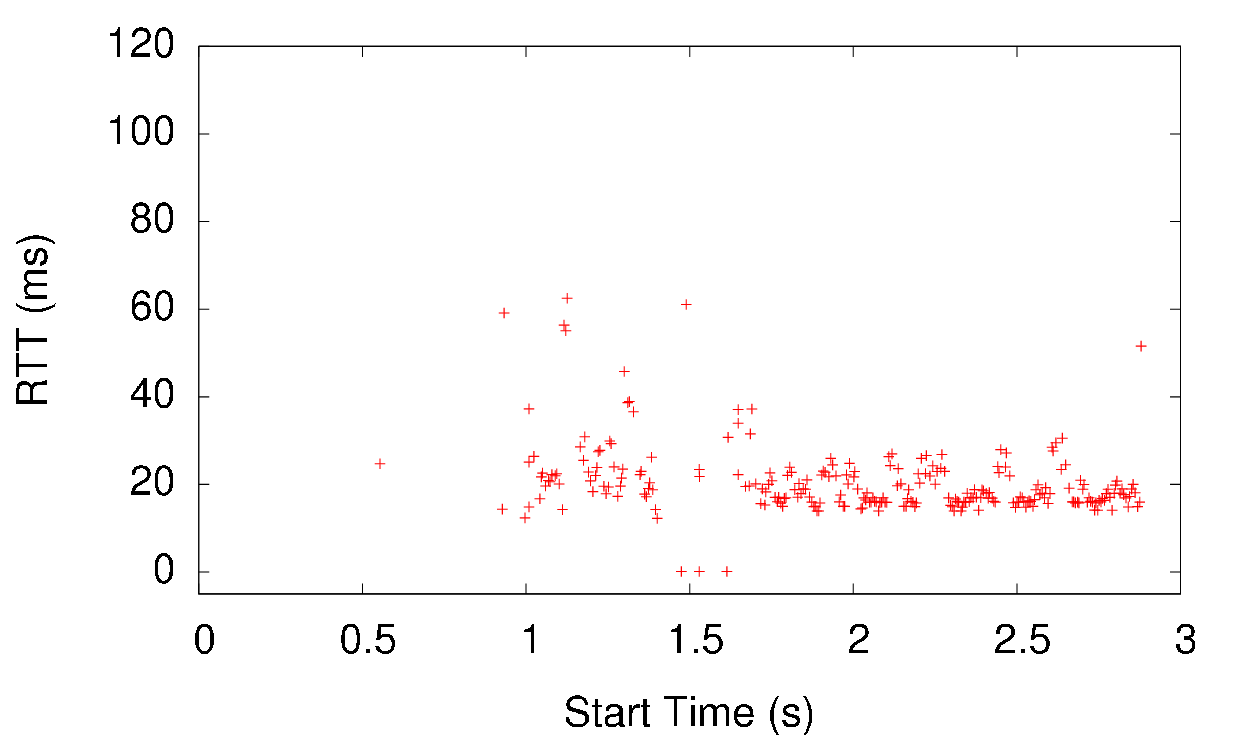}
\includegraphics[width=.70\textwidth]{rtt_over_time_graphs/anon_ccngetfile_rtt_over_time_test10MB}
\includegraphics[width=.70\textwidth]{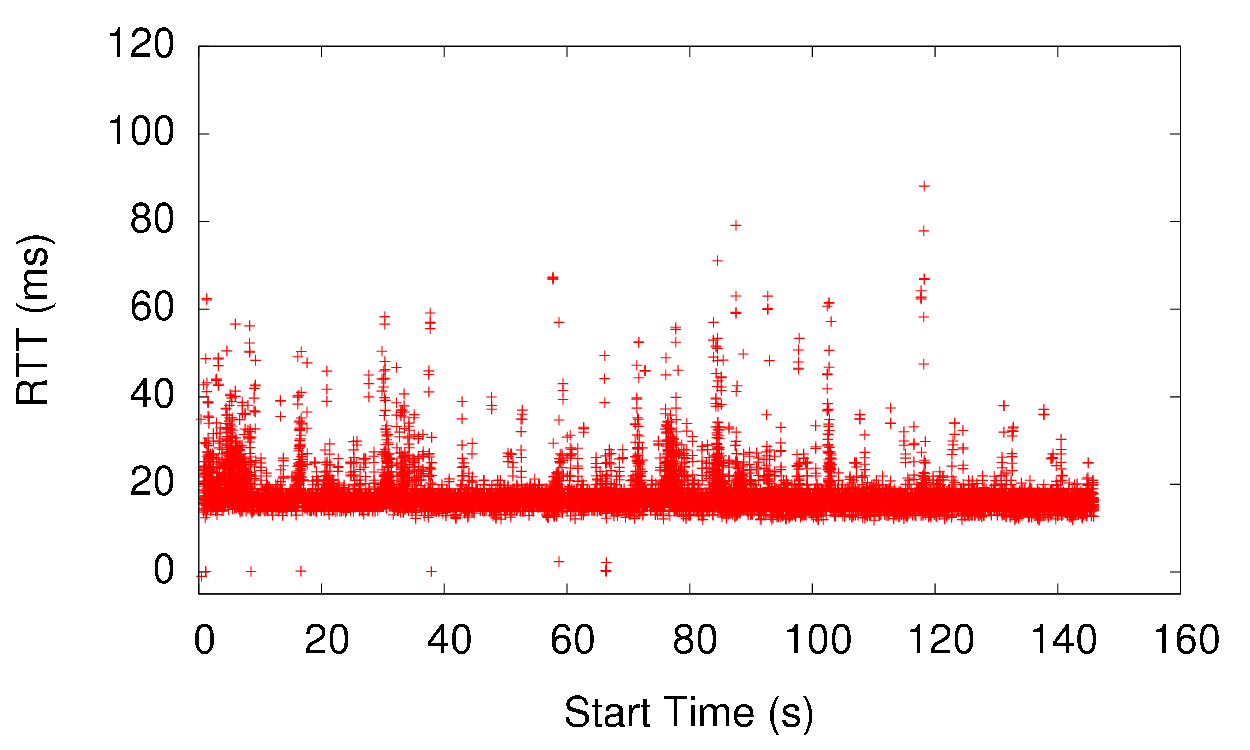}
\caption{\small Round trip time for transferring 1, 10 and 100MB of content over \ANON{} (full anonymity).}
\label{fig:andora-all}
\end{figure*}

\end{document}